\documentclass[a4paper,USenglish,cleveref, thm-restate]{lipics-v2021}

\hideLIPIcs  %

\title{Linear-Time MaxCut in Multigraphs Parameterized Above the \PT Bound}

\author{Jonas Lill}{Department of Computer Science, ETH Zürich, Switzerland}{jolill@student.ethz.ch}{}{}

\author{Kalina Petrova}{Institute of Science and Technology Austria, Austria\footnote{Parts of this research was conducted while Kalina Petrova was at the Department of Computer Science, ETH~Zürich, Switzerland.}}{Kalina.Petrova@ist.ac.at}{https://orcid.org/0009-0006-1753-6962}{Swiss National Science Foundation, grant no. CRSII5 173721. This project has received funding from the European Union’s Horizon 2020 research and innovation
programme under the Marie Skłodowska-Curie grant agreement No 101034413.}

\author{Simon Weber}{Department of Computer Science, ETH Zürich, Switzerland}{simon.weber@inf.ethz.ch}{https://orcid.org/0000-0003-1901-3621}{Swiss National Science Foundation under project no. 204320.}

\authorrunning{J. Lill, K. Petrova and S. Weber} %

\Copyright{Jonas Lill, Kalina Petrova and Simon Weber} %



\category{} %

\nolinenumbers %

\EventEditors{John Q. Open and Joan R. Access}
\EventNoEds{2}
\EventLongTitle{42nd Conference on Very Important Topics (CVIT 2016)}
\EventShortTitle{CVIT 2016}
\EventAcronym{CVIT}
\EventYear{2016}
\EventDate{December 24--27, 2016}
\EventLocation{Little Whinging, United Kingdom}
\EventLogo{}
\SeriesVolume{42}
\ArticleNo{23}

\usepackage[utf8]{inputenc}
\usepackage{amsmath,amssymb,amsfonts,mathrsfs,amsthm}
\usepackage[capitalise]{cleveref}

\crefname{claim}{claim}{claims}
\crefname{observation}{observation}{observations}

\usepackage{tikz}
\usepackage{xspace}

\newcommand{\erdos}{Erdős\xspace}
\newcommand{\EE}{Edwards-\erdos}
\newcommand{\turzik}{Turz\'{i}k\xspace}
\newcommand{\PT}{Poljak-\turzik}
\newcommand{\maxcutproblem}{\textsc{MaxCut}\xspace}
\newcommand{\maxcutwithweightsproblem}{\textsc{MaxCut-With-Vertex-Weights}\xspace}
\newcommand{\mac}{\ensuremath{\mu}}
\newcommand{\mst}{\ensuremath{w_{MSF}}}

\newcommand{\f}[1]{\relax\ifmmode#1\else{$#1$}\fi}

\newcommand{\R}{\mathbb{R}\xspace}
\newcommand{\N}{\mathbb{N}\xspace}

\usepackage{tikz}
\usepackage{subcaption}

\newcommand{\NP}{\ensuremath{\mathsf{NP}}}

\newcommand{\krulehead}{{\bf Rule \arabic{krule}}}
\crefname{krule}{rule}{rules}
\Crefname{krule}{Rule}{Rules}

\begin{document}

\maketitle

\begin{abstract}
\maxcutproblem is a classical \NP-complete problem and a crucial building block in many combinatorial algorithms. The famous \emph{\EE bound} states that any connected graph on $n$ vertices with $m$ edges contains a cut of size at least $\frac{m}{2}+\frac{n-1}{4}$. Crowston, Jones and Mnich [Algorithmica, 2015] showed that the \maxcutproblem problem on simple connected graphs admits an FPT algorithm, where the parameter $k$ is the difference between the desired cut size $c$ and the lower bound given by the \EE bound. This was later improved by Etscheid and Mnich [Algorithmica, 2017] to run in parameterized linear time, i.e., $f(k)\cdot O(m)$. We improve upon this result in two ways: Firstly, we extend the algorithm to work also for \emph{multigraphs} (alternatively, graphs with positive integer weights). Secondly, we change the parameter; instead of the difference to the \EE bound, we use the difference to the \emph{\PT bound}. The \PT bound states that any weighted graph $G$ has a cut of size at least $\frac{w(G)}{2}+\frac{\mst(G)}{4}$, where $w(G)$ denotes the total weight of $G$, and $\mst(G)$ denotes the weight of its minimum spanning forest. In connected simple graphs the two bounds are equivalent, but for multigraphs the \PT bound can be larger and thus yield a smaller parameter $k$. Our algorithm also runs in parameterized linear time, i.e., $f(k)\cdot O(m+n)$.
\end{abstract}

\newpage
\section{Introduction}
The $\maxcutproblem(G,c)$ problem is the problem of deciding whether a given graph $G$ contains a cut of size at least $c$. It has been known for a very long time that this problem is \NP-complete, in fact it was one of Karp's 21 \NP-complete problems~\cite{Karp1972}. The \maxcutproblem problem has been intensely studied from various angles such as random graph theory and combinatorics, but also approximation and parameterized complexity. It has numerous applications in areas such as physics and circuit design; for more background on the \maxcutproblem problem we refer to the excellent survey~\cite{maxcutsurvey}.

There are many lower bounds on the maximum cut size $\mac(G)$ of a given graph $G$. If $G$ is a graph with $m$ edges, a trivial lower bound is $\mac(G)\geq \frac{m}{2}$. 
This can be shown easily using the probabilistic method, as first done by \erdos~\cite{Erdös1965}. Clearly, $\maxcutproblem(G,c)$ is thus easily solvable if $c\leq\frac{m}{2}$. But what if $c$ is larger? At which point does the \maxcutproblem problem become difficult? It turns out that already $c=\frac{m}{2}+\epsilon m$ for any fixed $\epsilon>0$ makes the problem \NP-hard~\cite{NPhardmhalfpluseps}. However, as long as the difference $c-\frac{m}{2}$ is just a constant, $\maxcutproblem(G,c)$ is still polynomial-time solvable: Mahajan, Raman and Sikdar showed in 2009~\cite{firstFPT} that $\maxcutproblem(G,\frac{m}{2}+k)$ is fixed-parameter tractable (FPT), i.e., it can be solved in time $f(k)\cdot n^{O(1)}$. This started off the study of \emph{parameterized algorithms above guaranteed lower bounds}.

By the time this FPT algorithm was found, $\frac{m}{2}$ was no longer the best-known lower bound for $\mac(G)$. Already more than 40 years earlier, Edwards showed the following lower bound that was previously conjectured by \erdos, and is thus now known as the \EE bound.
\begin{theorem}[\EE bound~\cite{edwards1973some,edwards1975improved}]\label{lem:EEbound}
    For any connected simple graph $G$ with $n$ vertices and $m$ edges, $\mac(G)\geq \frac{m}{2}+\frac{n-1}{4}$.
\end{theorem}

Unlike the previous bound of $\frac{m}{2}$, this bound is tight for an infinite class of graphs, for example the odd cliques. It remained open for quite a while whether $\maxcutproblem(G,\frac{m}{2}+\frac{n-1}{4}+k)$ would also be fixed-parameter tractable, i.e., whether the parameter $k$ could be reduced by $\frac{n-1}{4}$ compared to the previous result by Mahajan et al. This question was answered in the positive by Crowston, Jones and Mnich, who proved the following theorem.

\begin{theorem}[Crowston, Jones, Mnich {\cite[Thm. 1]{blackboxFPT}}]\label{thm:fptKnown}
    There is an algorithm that computes, for any connected graph $G$ with $n$ vertices and $m$ edges and any integer $k$, in time $2^{O(k)}\cdot n^4$ a cut of $G$ of size at least $\frac{m}{2}+\frac{n-1}{4}+k$, or decides that no such cut exists.
\end{theorem}

This algorithm has later been improved to run in linear time (in terms of $m$) by Etscheid and Mnich~\cite{blackboxFPTlinear}. However, this improvement only holds for deciding the existence of such a cut, and not for computing a cut if one exists.

We would like to highlight another classic lower bound on the size of the maximum cut of a graph, nicknamed the ``spanning tree'' bound: Any connected graph on $n$ vertices has a cut of size at least $n-1$, since it contains a spanning tree of this size and trees are bipartite. Note that this bound is incomparable to the \EE bound. In 2018, Madathil, Saurabh, and Zehavi~\cite{aboveSpanningTree} showed that $\maxcutproblem(G,n-1+k)$ is also fixed-parameter tractable.

In 1986, Poljak and \turzik improved upon the \EE bound by replacing the term $n-1$ with the size of the minimum spanning tree (or forest in disconnected graphs), thus obtaining the following lower bound for maximum cuts in weighted graphs.
\begin{theorem}[\PT bound \cite{ptBound}]\label{lem:PTbound}
    For any graph $G=(V,E)$ with weight function $w:E\rightarrow\R_{>0}$, we have $\mac(G)\geq \frac{w(G)}{2}+\frac{\mst(G)}{4}$, where $w(G)=\sum_{e\in E}w(e)$ and $\mst(G)$ denotes the weight of a minimum-weight spanning forest of $G$.
\end{theorem}

It is easy to see that \Cref{lem:PTbound} implies the bound in \Cref{lem:EEbound} both for (unweighted) simple graphs and multigraphs. In unweighted simple graphs it is actually equivalent to \Cref{lem:EEbound}, while on multigraphs and positive integer-weighted graphs it can be strictly larger.

The authors of \Cref{thm:fptKnown} thus posed as their major open question whether their algorithm could be extended to solve $\maxcutproblem(G,\frac{m}{2}+\frac{n-1}{4}+k)$ on multigraphs as well. We answer this question in the positive, and improve the result further by replacing the \EE bound with the \PT bound.

\subsection{Results}
We provide a parameterized linear time algorithm for deciding \maxcutproblem in multigraphs and positive integer-weighted (simple) graphs above the \PT bound. A multigraph can be easily turned into a positive integer-weighted graph and vice versa; in the rest of this paper we phrase all of our results and proofs in terms of positive integer-weighted graphs for better legibility.
\begin{theorem}\label{thm:main}
    There is an algorithm that decides for any graph $G=(V,E)$ with weight function $w:E\rightarrow \N$ and any integer $k$, in time $2^{O(k)}\cdot O(|E|+|V|)$, whether a cut of $G$ of size at least $\frac{w(G)}{2}+\frac{\mst(G)}{4}+k$ exists.
\end{theorem}
Using the same techniques we can also get a parameterized quadratic-time algorithm to compute such a cut, if one exists.
\begin{theorem}\label{thm:finding}
    There is an algorithm that computes for any graph $G=(V,E)$ with weight function $w:E\rightarrow \N$ and any integer $k$, in time $2^{O(k)}\cdot O(|E|\cdot|V|)$, a cut of $G$ of size at least $\frac{w(G)}{2}+\frac{\mst(G)}{4}+k$, if one exists.
\end{theorem}

We would like to point out that \Cref{thm:main} is a strict improvement on the linear-time algorithm from \cite{blackboxFPTlinear} in two ways: Firstly we increase the types of graphs the algorithm is applicable to, and secondly we also strictly decrease the parameter for some instances. The following observation shows that this decrease of parameter can be significant.
\begin{observation}
    There exist sequences of positive-integer-weighted graphs $(G_i)_{i\in \N}$ and integers $(c_i)_{i\in \N}$ such that \Cref{thm:main} yields a polynomial-time algorithm to solve $\maxcutproblem(G_i,c_i)$, but when replacing $\mst(G)$ by $n-1$, it does not.
\end{observation}
\begin{proof}
    Let $G_i$ be a tree on $i+1$ vertices where each edge has weight $2$. Then, the \PT bound yields $\mac(G_i)\geq \frac{2i}{2}+\frac{2i}{4}=\frac{6}{4}i$, while the \EE bound only yields $\mac(G_i)\geq \frac{2i}{2}+\frac{i+1-1}{4}=\frac{5}{4}i$. Thus, if we set $c_i=\frac{6}{4}i+k$ for some constant $k$, then \Cref{thm:main} yields a $2^{O(k)}\cdot poly(i)=poly(i)$ algorithm for $\maxcutproblem(G_i,c_i)$, while with the \EE bound it would yield a $2^{O(k+\frac{1}{4}i)}\cdot poly(i)$ algorithm, which is not polynomial. 
\end{proof}

\subsection{Algorithm Overview}\label{sec:overview}
Our algorithm works in a very similar fashion to the one in \cite{blackboxFPT}. We use a series of reduction rules that can reduce the input graph down to a graph with no edges. While performing this reduction, we either prove that $G$ has a cut of the desired size, or we collect a set $S$ of $O(k)$ vertices such that $G-S$ is a uniform-clique-forest, i.e., a graph in which every biconnected component is a clique in which every edge has the same weight. Given such a set $S$, we can then compute the maximum cut of $G$ exactly: We iteratively test all possibilities of partitioning the vertices in $S$ between the two sides of the cut, and then compute the maximum cut of $G$ assuming that the vertices of $S$ are indeed partitioned like this. To do this, we use a similar approach as in \cite{blackboxFPT}: We compute the maximum cut of $G-S$ with \emph{weighted vertices}. In this setting, each vertex $v$ in $G-S$ specifies a weight $w_0(v)$ and $w_1(v)$ for both possible sides of the cut $v$ may land in. The value of a cut is given by the total weight of the cut edges plus the sum of the correct weight for each vertex. To use this problem to compute the maximum cut of $G$, we set the weights of each vertex $v$ in $G-S$ according to the total weight of the edges between $v$ and $S$ that are cut in the assumed partition of $S$. Maximizing over all possible partitions for $S$ gives the maximum cut of $G$.

While we use very similar techniques as in \cite{blackboxFPT,blackboxFPTlinear}, our main technical contribution lies in the reduction rules. Our reduction rules have to be more specific, i.e., each reduction rule has a stronger precondition. This is due to the fact that when performing any reduction, the change in the weight of a minimum spanning forest (as needed for the \PT bound) is much more difficult to track than the number of vertices in the graph (as needed for the \EE bound). Since our rules are more specific, we also need twice as many rules as in \cite{blackboxFPT} (and one more rule than \cite{blackboxFPTlinear}) to ensure that always at least one rule is applicable to a given graph.

\section{Preliminaries}

In the rest of this paper we consider every graph to be a simple graph $G=(V,E)$, where $V$ is the set of vertices, and $E\subseteq \binom{V}{2}$ is the set of edges. A graph is weighted if it is equipped with a positive integer edge-weight function $w:E\rightarrow \N$. For any two disjoint subsets $A,B\subseteq V$ we denote by $E(A,B)$ the set of edges between $A$ and $B$, by $w(A,B)$ the total weight of the edges in $E(A,B)$, and by $\min(A,B)$ the minimum weight of any edge in $E(A,B)$. For a subset $A \subseteq V$, we denote by $N(A)$ the set of vertices in $V \setminus A$ that have a neighbor in $A$.

A \emph{cut} is a subset $C\subseteq V$, and the \emph{weight} of a cut $C$ is the total weight of the edges connecting a vertex in $C$ to a vertex in $V\setminus C$, i.e., $w(C)=w(C,V\setminus C)$.

For any set $A\subseteq V$ we write $G[A]$ for the graph on $A$ induced by $G$, and $G-A$ for the graph on $V\setminus A$ induced by $G$.

We say that a graph is \emph{uniform} if all of the edges have the same weight. More specifically, we call a graph $c$-uniform if all edges have weight $c$.

A graph $(V,E)$ is called \emph{biconnected}, if $|V|\geq 1$, and for every vertex $v\in V$, $G-\{v\}$ is connected. A \emph{biconnected component} of a graph is a maximal biconnected subgraph, also referred to as a \emph{block}. It is well-known that the biconnected components of every graph partition its edges. A vertex that participates in more than one biconnected component is a \emph{cut vertex} (usually defined as a vertex whose removal disconnects a connected component). A graph can thus be decomposed into biconnected components and cut vertices.
\begin{definition}[Block-Cut Forest]
    The \emph{block-cut forest} $F$ of a graph $G$ has vertex set $V(F) = \mathcal{C} \cup \mathcal{B}$, where $\mathcal{C}$ is the set of cut vertices of $G$ and $\mathcal{B}$ is the set of biconnected components of $G$, and $\{B,c\}$ is an edge in $F$ if $B \in \mathcal{B}$, $c \in \mathcal{C}$, and $c \in V(B)$. 
\end{definition}
It is not hard to see that the block-cut forest $F$ of a graph $G$ is indeed a forest, since a cycle in it would imply a cycle in $G$ going through multiple biconnected components, thus contradicting their maximality. Moreover, each connected component of $F$ corresponds to a connected component of $G$, and all leaves of $F$ are biconnected components in $G$. We refer to the biconnected components of $G$ that correspond to leaves of $F$ as \emph{leaf-blocks} of $G$. 

\begin{definition}[Uniform-Clique-Forest]
    A weighted graph is a \emph{uniform-clique-forest} if each of its blocks $B$ is a uniform clique.
\end{definition}

\begin{definition}
    The problem \maxcutwithweightsproblem is given as follows.
    \begin{description}
        \item[Input:] A weighted graph $(V,E)$ with edge-weight function $w$, as well as two vertex-weight functions $w_0:V\rightarrow \N$, $w_1:V\rightarrow \N$.
        \item[Output:] A cut $C$ maximizing $w(C)+\sum_{v\in C}w_1(v)+\sum_{v\not\in C}w_0(v)$.
    \end{description}
\end{definition}

We show in \Cref{sec:onuniform} that \maxcutwithweightsproblem is solvable in linear time if the input graph is a uniform-clique-forest.

\section{Reducing to a Uniform-Clique-Forest}
In the first part of our algorithm, we wish to either already conclude that the input graph has a cut of the desired size, or to find some set $S$ of vertices such that $G-S$ is a uniform-clique-forest.
\begin{lemma}\label{lem:reduce}
    For any graph $G=(V,E)$ on $n$ vertices with $m$ edges and weight function $w:E\rightarrow \N$ and any integer $k$, in time $O(n+k\cdot m)$ one can either decide that $G$ has a cut of size at least $\frac{w(G)}{2}+\frac{\mst(G)}{4}+\frac{k}{4}$, or find a set $S\subseteq V$ such that $|S|\leq 3k$ and $G-S$ is a uniform-clique-forest.
\end{lemma}
Note that we write $\frac{k}{4}$ instead of just $k$. The reason for this is that with our reduction rules we make ``progress'' reducing the difference to the \PT bound in increments of $\frac{1}{4}$.

To prove \Cref{lem:reduce} we use eight reduction rules, closely inspired by the reduction rules used in \cite{blackboxFPT,blackboxFPTlinear}. Each reduction rule removes some vertices from the given graph, possibly \emph{marks} some of the removed vertices to be put into $S$, and possibly \emph{reduces} the parameter $k$ by $1$. To prove \Cref{lem:reduce}, the reduction rules will be shown to fulfill the following properties.

Firstly, each reduction rule ensures a one-directional implication: if the reduced graph $G'$ contains a cut of size  $\frac{w(G')}{2}+\frac{\mst(G')}{4}+\frac{k'}{4}$ (where $k'$ is the possibly reduced $k$), then the original graph $G$ must also contain a cut of size $\frac{w(G)}{2}+\frac{\mst(G)}{4}+\frac{k}{4}$. By the \PT bound, if $k$ ever reaches $0$, it is clear that the original graph $G$ must have contained a cut of the desired size.

Secondly, we need that to every graph with at least one edge, at least one of the rules applies. To get our desired runtime, we also need that an applicable rule can be found and applied efficiently.

Thirdly, every rule should only mark at most three vertices to be added to $S$. If a rule does not reduce $k$, it may not mark any vertices. This ensures that at most $3k$ vertices are added to~$S$.

Lastly, we require that after exhaustively applying the rules and reaching a graph with no more edges, the graph $G-S$ is a uniform-clique-forest.

We will now state our reduction rules, and then prove these four properties in \Cref{lem:rulesaresound,lem:rulecanbeapplied}, \Cref{obs:markthree}, and \Cref{lem:GminusSuniformCF}, respectively. For simplicity, each reduction rule is stated in such a way that it assumes the input graph to be connected. If the input graph is disconnected, instead consider $G$ to be one of its connected components. Each rule preserves connectedness of the connected component it is applied to, which we also show in \Cref{lem:rulesaresound}. Note further that if the connected component the rule is being applied to is also biconnected, then if the precondition requires some vertex to be a cut vertex, any vertex can play that role, although technically there are no cut vertices. We state this once here for simplicity, instead of saying each time that either $v$ is a cut vertex or $G$ is biconnected. We visualize the eight rules in \Cref{fig:rules}.

\refstepcounter{krule}
\label{rule:edge}
\noindent
\framebox[\textwidth][l]{
    \begin{tabularx}{0.98\textwidth}{lX}
        \krulehead: &Let $\{x,y\},\{y,z\} \in E$ be such that $w(x,y) > w(y,z)$ and $G - \{x,y\}$ is connected.\\
        Remove: &$\{x,y\}$
        \\Mark: &$\{x,y\}$
        \\Reduce $k$: &Yes
    \end{tabularx}
}

\smallskip
\refstepcounter{krule}
\label{rule:vertexclique_constant}
\noindent
\framebox[\textwidth][l]{
    \begin{tabularx}{0.98\textwidth}{lX}
        \krulehead: &Let $X \subseteq V$, $v \in V \setminus X$ be such that $X \cup \{v\}$ is a leaf-block of $G$ with cut vertex $v$, and $G[X\cup\{v\}]$ is a uniform clique.
        \\Remove: &$X$
        \\Mark: &$\emptyset$
        \\Reduce $k$: &No
    \end{tabularx}
}

\smallskip
\refstepcounter{krule}
\label{rule:vertexclique}
\noindent
\framebox[\textwidth][l]{
    \begin{tabularx}{0.98\textwidth}{lX}
        \krulehead: &Let $X \subseteq V$, $v \in V \setminus X$ be such that $X \cup \{v\}$ is a clique and a leaf-block of $G$ with cut vertex $v$; $G[X]$ is uniform, and $G[X\cup\{v\}]$ is not uniform.
        \\Remove: &$X$
        \\Mark: &$\{v\}$
        \\Reduce $k$: &Yes
    \end{tabularx}

}
 
\smallskip
\refstepcounter{krule}
\label{rule:vertexmark}
\noindent
\framebox[\textwidth][l]{
    \begin{tabularx}{0.98\textwidth}{lX}
        \krulehead: &Let $X \subseteq V$, $v \in V \setminus X$ be such that $X \cup \{v\}$ is a leaf-block of $G$ with cut vertex $v$; $v$ has at least two neighbors in $X$; $G[X]$ is a uniform clique; $G[X \cup \{v\}]$ is not a clique.
        \\Remove: &$X$
        \\Mark: &$\{v\}$
        \\Reduce $k$: &Yes
    \end{tabularx}
}

\smallskip
\refstepcounter{krule}
\label{rule:specialvertexmark}
\noindent
\framebox[\textwidth][l]{
    \begin{tabularx}{0.98\textwidth}{lX}
        \krulehead: &Let $X \subseteq V$, $v \in V \setminus X$ be such that $X \cup \{v\}$ is a leaf-block of $G$ with cut vertex $v$; $G[X]$ is a clique; $v$ has exactly two neighbors $x,y$ in $X$; all edges in $G[X]$ have weight~$c$, except $\{x,y\}$, which has weight $w(x,y) > c$; $w(v,x),w(v,y)\geq c$.
        \\Remove: &$X$
        \\Mark: &$\{v,x,y\}$
        \\Reduce $k$: &Yes
    \end{tabularx}
}

\smallskip
\refstepcounter{krule}
\label{rule:triplet}
\noindent
\framebox[\textwidth][l]{
    \begin{tabularx}{0.98\textwidth}{lX}
        \krulehead: &Let $a,b,c \in V$ be such that $\{a,b\},\{b,c\} \in E$; $\{a,c\}\notin E$; $G - \{a,b,c\}$ is connected; $w(a,b) = w(b,c)$; and $2w(a,b) > \min(\{a,b,c\},V\setminus\{a,b,c\})$.
        \\Remove: &$\{a,b,c\}$
        \\Mark: &$\{a,b,c\}$
        \\Reduce $k$: &Yes
    \end{tabularx}
}

\smallskip
\refstepcounter{krule}
\label{rule:specialtriplet}
\noindent
\framebox[\textwidth][l]{
    \begin{tabularx}{0.98\textwidth}{lX}
        \krulehead: &Let $v,a,b,c$ $\in$ $V$ be such that $\{a,b,c,v\}$ is a leaf-block of $G$ with cut vertex $v$; $\{a,b\},\{b,c\}, \{a,v\}, \{c,v\}\in E$; $\{a,c\}\notin E$; $w(a,b) = w(b,c)$; $w(a,v), w(c,v) \geq 2w(a,b)$; and if $\{b,v\} \in E$ then $w(b,v) \geq 2w(a,b)$.
        \\Remove: &$\{a,b,c\}$
        \\Mark: &$\{a,b,c\}$
        \\Reduce $k$: &Yes
    \end{tabularx}
}

\smallskip
\refstepcounter{krule}
\label{rule:specialxy}
\noindent
\framebox[\textwidth][l]{
    \begin{tabularx}{0.98\textwidth}{lX}%
        \krulehead: &Let $x,y\in V$ be such that $\{x,y\}\notin E$; $G - \{x,y\}$ has exactly two connected components $X$ and $Y$; $G[X\cup \{x\}]$ and $G[X \cup \{y\}]$ are both $c$-uniform cliques; and $x$ and $y$ have exactly one neighbor $v$ in $Y$.
        \\Remove: &$X \cup \{x,y\}$
        \\Mark: &$\{x,y\}$
        \\Reduce $k$: &Yes
    \end{tabularx}
}
\smallskip

\begin{figure}[htb!]
    \centering
    \includegraphics[]{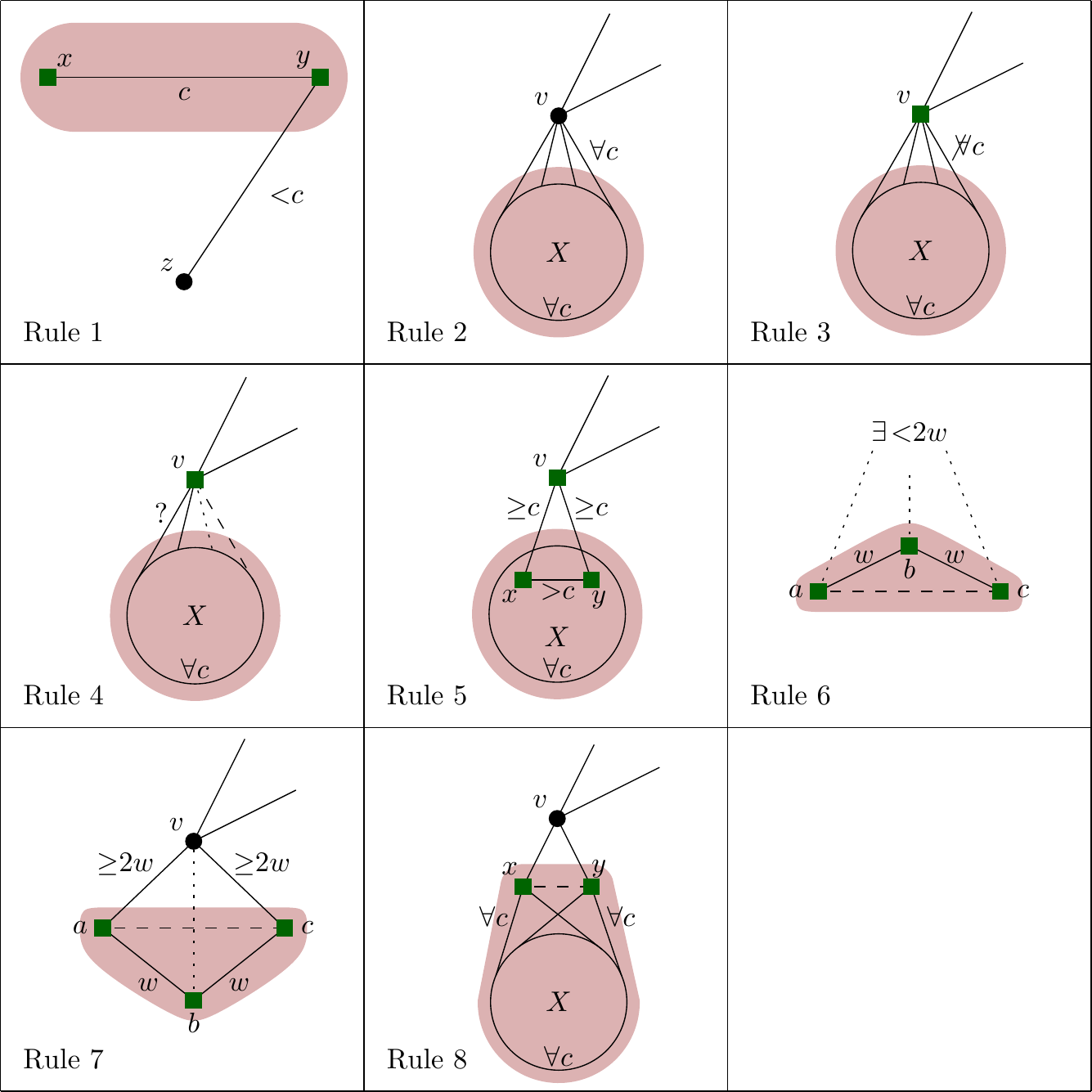}
    \caption{The eight reduction rules. An edge is drawn normally if it must exist for the rule to apply. Some edges are drawn dashed to emphasize that they \emph{must not} exist for the rule to apply. Some additional edges are drawn dotted to emphasize that they \emph{may} exist but do not have to. Red shading indicates the vertices removed by the rule, while vertices marked by the rule are drawn using a green square.}
    \label{fig:rules}
\end{figure}

We first state the formalizations of our four properties, then prove \Cref{lem:reduce}, and only then prove each of our properties.

\begin{lemma}\label{lem:rulesaresound}
    Let $G=(V,E)$ be a graph with weights $w$, and let $k$ be any positive integer. Let $G'$ be the result of one application of one of the rules 1--8 to $G$, and $k'$ the resulting parameter. Then, if $G'$ has a cut of size at least $\frac{w(G')}{2}+\frac{\mst(G')}{4}+\frac{k'}{4}$, then $G$ must contain a cut of size at least $\frac{w(G)}{2}+\frac{\mst(G)}{4}+\frac{k}{4}$. Furthermore, if $G$ is connected, then $G'$ is connected.
\end{lemma}

\begin{lemma}\label{lem:rulecanbeapplied}
    Let $G=(V,E)$ be a weighted graph with at least one edge. Given the block-cut forest of $G$ we can either apply \Cref{rule:vertexclique_constant} in time $O(|E'|)$ where $E'$ is the set of edges removed by applying \Cref{rule:vertexclique_constant}, or we can find and apply another rule in time $O(|E|)$. In the same time we can also adapt the  block-cut forest.
\end{lemma}

\begin{observation}\label{obs:markthree}
    Each rule marks at most three vertices. \Cref{rule:vertexclique_constant}, the only rule that does not reduce $k$, does not mark any vertices.
\end{observation}

\begin{lemma}\label{lem:GminusSuniformCF}
     Let $S$ be the set of vertices marked when exhaustively (i.e., until $G$ has no edges) applying Rules 1--8 to a graph $G$. Then $G - S$ is a uniform-clique-forest.
\end{lemma}

Let us now prove \Cref{lem:reduce} using these properties.

\begin{proof}[Proof of \Cref{lem:reduce}]
    We begin by computing the block-cut forest of $G$ in $O(n+m)$ time~\cite{biconnectedcomponents}. Then, we apply rules until we either reach $k=0$ or until we reach a graph with no edges. Whenever we apply a rule, we locally adapt the block-cut forest. In total we apply rules other than \Cref{rule:vertexclique_constant} at most $k$ times. By \Cref{lem:rulecanbeapplied} this takes at most $O(k\cdot m)$ time. Since applying \Cref{rule:vertexclique_constant} takes time $O(|E'|)$ where $E'$ is the set of edges removed, all applications of \Cref{rule:vertexclique_constant} together use time $O(m)$. The reduction step can thus be performed in $O(k\cdot m)$.

    If we have reached $k=0$, by the \PT bound and by \Cref{lem:rulesaresound} we can decide that our input graph contains a cut of the desired size. Otherwise, by \Cref{obs:markthree}, $S$ contains at most $3k$ vertices. By \Cref{lem:GminusSuniformCF}, $G-S$ then forms a uniform-clique-forest, and we have proven our desired statement.
\end{proof}

We will now proceed to prove \Cref{lem:rulesaresound,lem:rulecanbeapplied,lem:GminusSuniformCF}.
The main technical challenges are the proofs of \Cref{lem:rulesaresound,lem:rulecanbeapplied}. These proofs are more technically involved than the corresponding proofs from \cite{blackboxFPTlinear}. For \Cref{lem:rulesaresound} this is due to the fact that the weight of a minimum spanning forest is much more difficult to track through a reduction than the number of vertices. For \Cref{lem:rulecanbeapplied} the proof is more involved since our rules are more specific, and thus more case distinction is needed. We present the proof of \Cref{lem:rulesaresound} in \Cref{app:proofsoundness}, since despite its technicality, it is not very insightful.

To prove \Cref{lem:rulecanbeapplied} we use the following lemma, the proof of which follows from the proof of \cite[Lemma 3]{blackboxFPTlinear} rather directly.
\begin{lemma}
\label{lemma:structural_properties}
Let $G=(V,E)$ be a connected graph with at least one edge, and let $B\subseteq V$ be a biconnected component that is a leaf in the block-cut forest of $G$. Now, we write $B$ as $X\cup\{v\}$, where $v$ is the cut vertex disconnecting $B=X\cup\{v\}$ from $V\setminus B$ (if $B$ is an isolated vertex in the block-cut forest, i.e., it forms a connected component of $G$ that is also biconnected, then let $v$ be an arbitrary vertex in $B$). Then at least one of the following properties holds.
\begin{enumerate}
\item[A)] $G[X\cup\{v\}]$ is a clique.
\item[B)] $G[X]$ is a clique but $G[X\cup\{v\}]$ is not a clique.
\item[C)] $v$ has exactly two neighbors in $X$, $x$ and $y$. Furthermore, $\{x,y\} \notin E$, and $G[X\setminus\{x\}]$ and $G[X\setminus\{y\}]$ are cliques.
\item[D)] $X\cup\{v\}$ contains vertices $a,b,c$ such that $\{a,b\},\{b,c\} \in E$, $\{a,c\}\notin  E$, and $G - \{a,b,c\}$ is connected.
\end{enumerate}
Furthermore, such a property (including the vertices $x,y$ and $a,b,c$ for cases C and D, respectively) can be found in linear time in the number of edges in $G[X]$.
\end{lemma}
\begin{proof}[Proof (sketch)]
    One can check whether $G[X]$ is a clique for some $X\subseteq V$ in time linear in the number of edges in $G[X]$. To do this, we simply check whether each edge is present in some fixed order. It is thus easy to check for cases A), B), and C) in linear time.

    In the proof of \cite[Lemma 3]{blackboxFPTlinear} it is shown that if none of the cases A), B), and C) apply, then vertices $a,b,c$ certifying case D) can be found in linear time.
\end{proof}

\begin{proof}[Proof of \Cref{lem:rulecanbeapplied}]
    Without loss of generality we can assume that $G$ is connected; otherwise, we consider $G$ to be an arbitrary connected component of our input graph that contains at least one edge. We first apply \Cref{lemma:structural_properties} on a leaf-block $X\cup\{v\}$ to find one of the four properties.
    
    \proofsubparagraph{Property A)} If property A) holds, we can check whether $G[X\cup\{v\}]$ is uniform in time $O(|E'|)$ where $E'$ is the set of edges in $G[X\cup\{v\}]$. In this process we can track also whether $G[X]$ is uniform. If $G[X\cup\{v\}]$ is uniform we apply \Cref{rule:vertexclique_constant}. Else, if only $G[X]$ is uniform, we apply \Cref{rule:vertexclique}. If not even $G[X]$ is uniform, we can find two edges $\{x,y\},\{y,z\}$ in $G[X]$ such that $w(x,y)>w(y,z)$. Since $X\cup\{v\}$ is a clique, $G-\{x,y\}$ must be connected. We can therefore apply \Cref{rule:edge}.

    \proofsubparagraph{Property B)} We can handle property B) in a similar way. If $G[X]$ is uniform, we can apply \Cref{rule:vertexmark}. Else, we apply case distinction on the number of vertices in $X$ adjacent to $v$. We first consider the case if vertex $v$ is adjacent to exactly two vertices in $X$.  Since $X$ is not uniform, there exist vertices $x,y \in X$ and a vertex $u\in X\cup\{v\}$ such that $w(x,y) > w(x,u)$. If the only such choice of $x,y$ is such that $x$ and $y$ are exactly the two vertices in $X$ adjacent to $v$, then we can apply \Cref{rule:specialvertexmark}. Else we can see that $G - \{x,y\}$ must be connected and apply \Cref{rule:edge}.
    Let us now consider the other case, that vertex $v$ is adjacent to at least three vertices in $X$. There must again exist vertices $u,x,y \in X$ so that $w(x,y) > w(x,u)$. Since $v$ is adjacent to at least three vertices and $G[X]$ is a clique, $G - \{x,y\}$ is connected and we can apply \Cref{rule:edge}.

    \proofsubparagraph{Property C)} To handle Property C) we first check whether $G[X]$ is uniform. If it is not, we can apply \Cref{rule:edge}, since for any edge $\{a,b\}$ in $G[X]$, $G-\{a,b\}$ is connected.  Knowing that $G[X]$ is uniform, and that $v$ has exactly two neighbors, we can apply \Cref{rule:specialxy}.
    
    \proofsubparagraph{Property D)} Note that since $G - \{a,b,c\}$ is connected, and since by its biconnectedness $B\neq\{a,b,c\}$, if $G - B$ is non-empty, then $v \notin \{a,b,c\}$. Next, again since $G[B]$ is biconnected, we must have that $E(\{a\},B\setminus\{a,b,c\})\neq\emptyset$ and $E(\{c\},B\setminus\{a,b,c\})\neq\emptyset$. From this we get that $G-\{a,b\}$ and $G-\{b,c\}$ must be connected. Thus, we can compare $w(a,b)$ and $w(b,c)$, and apply \Cref{rule:edge} if $w(a,b)\neq w(b,c)$. We can also compute the value $m=\min(\{a,b,c\},B\setminus\{a,b,c\})$. If $2w(a,b)>m$ we can apply \Cref{rule:triplet}.

    Next, we compute the block-cut forests for the four graphs $G_{abc}:=G[B]-\{a,b,c\}$ and $G_u:=G[B]-\{u\}$ for all $u\in\{a,b,c\}$. This can be performed in the required time, and yields the set of cut vertices for all these graphs. We now test for every $u\in\{a,b,c\}$ and for every vertex $z\in B \setminus \{a,b,c\}$ with $\{z,u\}\in E$ whether $z$ is a cut vertex in $G_u$. If for any such pair $z,u$ we have that $z$ is \emph{not} a cut vertex in $G_u$, this means that $G-\{u,z\}$ is connected, and we can thus apply \Cref{rule:edge} to that edge (recall that since we could not apply \Cref{rule:triplet} earlier, every edge in $E(\{a,b,c\},B\setminus\{a,b,c\})$ has weight at least twice as large as $w(a,b)=w(b,c)$).
    
    If every such $z$ adjacent to some $u\in\{a,b,c\}$ \emph{is} a cut vertex in $G_u$, we check whether any of these vertices is \emph{not} a cut vertex in $G_{abc}$. If one is not, we claim that we can apply \Cref{rule:specialtriplet}. We prove this by distinguishing two cases, depending on $u$:
    \begin{itemize}
        \item $u\in\{a,c\}$: Suppose without loss of generality that $u=a$. Since $z$ is a cut vertex of $G_a$, it follows that $G_a - z$ has $t\geq 2$ connected components $C_1, \dots, C_t$. Suppose without loss of generality that $b,c \in C_1$. If $C_1 \setminus \{b,c\} \neq \varnothing$, then $C_1 \setminus \{b,c\}$ and $C_2$ are two different connected components of $G_{abc} - z$, contradicting our assumption that $z$ is not a cut vertex of $G_{abc}$. Thus $C_1 = \{b,c\}$, implying that $b$ and $c$ have no neighbours in $B \setminus \{a,b,c,z\}$. Therefore $\{c,z\} \in E$ as $E(\{c\}, B \setminus \{a,b,c\}) \neq \varnothing$, so $c$ can also play the role of $u$. By symmetry, $a$ has no neighbours in $B \setminus \{a,b,c,z\}$ and $\{a,z\} \in E$. It follows that $\{a,b,c,z\}$ is a leaf-block of $G$ and so \Cref{rule:specialtriplet} applies.

        \item $u=b$: Let $S := B \setminus \{a,b,c,z\}$. Recall that we established that $E(\{a\},S \cup \{z\})$ and $E(\{c\}, S \cup \{z\})$ are both non-empty. We will show that either $\{a,z\} \in E$ or $\{c,z\} \in E$ (or both hold). Suppose that is not the case. Then $E(\{a\},S)$ and $E(\{c\}, S)$ are both non-empty. Since $z$ is not a cut vertex in $G_{abc}$, the graph $G[S]$ must be connected. That implies $G[S \cup \{a,c\}] = G[B] - \{b,z\}$ is also connected, which contradicts our assumption that $z$ is a cut vertex of $G_b$.

        We have shown that at least one of $\{a,z\}$ and $\{c, z \}$ is in $E$, say $\{a,z\}$. Thus, without loss of generality, the case $u \in \{a,c\}$ applies, since $z$ must be a cut vertex of $G_a$ by our assumption that this holds for all adjacent pairs $u,z$ with $u\in\{a,b,c\}$ and $z\in B \setminus \{a,b,c\}$. We have thus reduced the case $u=b$ to $u \in \{a,c\}$, which we already handled.

    \end{itemize}
    One can now show that if this point is reached without having found an applicable rule, then \Cref{rule:specialtriplet} must be applicable to the graph. Let us collect all the properties we know to be true (under the assumption that we have not found an applicable rule until now). \begin{enumerate}
    \item\label{prop1:rulecanbeapplied} $E(\{a\},B\setminus\{a,b,c\})\neq\emptyset$ and $E(\{c\},B\setminus\{a,b,c\})\neq\emptyset$.\item\label{prop2:rulecanbeapplied} $w(a,b) = w(b,c)$ \item\label{prop3:rulecanbeapplied} $2w(a,b) \leq \min (B- \{a,b,c\}, \{a,b,c\})$\item\label{prop4:rulecanbeapplied} For every pair of vertices $z \in B \setminus \{a,b,c\}$ and $u \in \{a,b,c\}$ with $\{z,u\} \in E$, $z$ is a cut vertex of both $G_u$ and  $G_{abc}$.
    \end{enumerate}

    Observe that since $B$ is biconnected containing at most one cut vertex of $G$, it follows that there can be at most one cut vertex of $G$ with a neighbor in $\{a,b,c\}$. We will now use the following claim that we will prove later.

    \begin{claim}\label{lemma:cut_vertices}
        Let $G$ be a connected graph with $X \subset G$ where $X$ and $G - X$ are connected, and for every vertex $v \in V(G - X)$, if $v$ has a neighbor in $X$, then $v$ is a cut vertex of $G - X$. If $|N(X)| \geq 2$, then there are two distinct vertices $v_1,v_2 \in N(X)$ that are both cut vertices of $G$.
    \end{claim}
    
    We apply \Cref{lemma:cut_vertices} on the set $X:=\{a,b,c\}$. By property~\ref{prop4:rulecanbeapplied} above, and by the fact that $N(\{a,b,c\})$ contains at most one cut vertex of $G$, we get that $|N(\{a,b,c\})| = 1$. The vertex in $N(\{a,b,c\})$ must be the cut vertex $v$. By property~\ref{prop1:rulecanbeapplied} we know that $\{a,v\}, \{c,v\} \in E$. By properties~\ref{prop2:rulecanbeapplied} and~\ref{prop3:rulecanbeapplied} all the weight restrictions of \Cref{rule:specialtriplet} are satisfied, which can thus be applied.
\end{proof}

\begin{proof}[Proof of \Cref{lemma:cut_vertices}] 
Let $H$ be the block-cut forest of $G-X$ and suppose $V(H) = \mathcal{C} \cup \mathcal{B}$, where $\mathcal{C}$ are the cut vertices of $G-X$ and $\mathcal{B}$ are the biconnected components of $G-X$. Since $|N(X)| \geq 2$, we get that $|\mathcal{C}| \geq 2$. Note that all leaves of $H$ are in $\mathcal{B}$. Consider the tree $T$ that we obtain by removing all leaves of $H$, and note that $T$ has at least two vertices since $\mathcal{C} \subseteq V(T)$. Thus, $T$ has at least two leaves, say $\ell_1, \ell_2$, each of which must be in $\mathcal{C}$, since its neighbors in $H\setminus T$ are in $\mathcal{B}$. Let $B' \in \mathcal{B}$ be a leaf of $H$ that is a neighbor of $\ell_i$ for some $i\in\{1,2\}$. Since every vertex in $N(X)$ is in $\mathcal{C}$, it follows that $E(X, B') = \varnothing$, so $\ell_i$ is a cut vertex in $G$.
\end{proof}

For this section, it only remains to prove \Cref{lem:GminusSuniformCF}.

\begin{proof}[Proof of \Cref{lem:GminusSuniformCF}]
    Let $G_1,G_2,\ldots,G_q$ be the sequence of graphs obtained while exhaustively applying rules 1--8 to $G_1$ ($G_2$ is the graph obtained after applying one rule to $G_1$, $G_3$ is the graph obtained after applying one rule to $G_2$, and so on). We prove that for any graph $G_i$ in the sequence, $G_i-S$ is a uniform-clique-forest. We run this proof by induction over the sequence of graphs in reverse order (in the order $G_q,G_{q-1}$,\ldots,$G_2$,$G_1$).
\\\textit{Base Case:} By \Cref{lem:rulecanbeapplied}, we know that $G_q$ is a graph without edges, therefore $G_q = G_q - S$ is trivially a uniform-clique-forest.
\\\textit{Induction Hypothesis: } Assume $G_i - S$ is a uniform-clique-forest.
\\\textit{Step Case:} We prove that $G_{i-1} - S$ is a uniform-clique-forest. We know that one rule among rules 1--8 was applied to $G_{i-1}$ to obtain $G_i$. We do a case distinction over which rule was applied:
\begin{itemize}
    \item \Cref{rule:edge}, \ref{rule:triplet}, or \ref{rule:specialtriplet} was applied to $G_{i-1}$. Every vertex these rules remove is also marked, therefore $G_{i-1} - S = G_i - S$.
    \item \Cref{rule:vertexclique_constant} was applied to $G_{i-1}$. We can create $G_{i-1}-S$ from $G_i-S$ by connecting a clique $X$ to a vertex $v \in V(G_i)$ such that $X \cup \{v\}$ is a uniform clique. If $v$ is in $S$, this is instead adding a disjoint uniform clique. Observe that this just adds a uniform leaf-clique in either case.
    \item \Cref{rule:vertexclique},  \ref{rule:vertexmark}, \ref{rule:specialvertexmark}, or \ref{rule:specialxy} was applied to $G_{i-1}$. We can create $G_{i-1}-S$ from $G_i-S$ by adding a disjoint uniform clique.
\end{itemize}
We conclude that in all cases $G_{i-1} - S$ consists of one or zero uniform cliques added to $G_i - S$ as a leaf, and thus by the induction hypothesis $G_{i-1}-S$ is a uniform-clique-forest.
\end{proof}

\section{Solving \maxcutwithweightsproblem on Uniform-Clique-Forests}\label{sec:onuniform}

\begin{lemma}\label{lem:onuniformCF}
\maxcutwithweightsproblem on a uniform-clique-forest $G$ with $n$ vertices and $m$ edges can be solved in $O(n + m)$ time. 
\end{lemma}

\begin{proof}
This proof loosely follows the proof of \cite[Lemma 4]{blackboxFPTlinear}. We first compute the cut-block forest of $G$. We know that every graph contains at least one leaf-block. Let $X \cup \{v\}$ be a leaf-block of $G$ where $v \in V(G)$ is the cut vertex of $X$ (if a connected component of $G$ consists of a single biconnected component $B$,  then $X = B - \{v\}$ where $v$ is an arbitrary vertex in $B$). Let $n' = |X|$ and $m'$ be the number of edges in $G[X \cup \{v\}]$. Since $G$ is a uniform-clique-forest, we know that $G[X\cup \{v\}]$ is $c$-uniform for some $c$. We now consider the maximum weighted cut in $G[X\cup\{v\}]$ for both possible cases $v\not\in C$ and $v\in C$.

We first consider $v\not\in C$. Let $\delta(x) = w_1(x) - w_0(x)$ for every vertex $x \in X$. We can sort the vertices in $X$ in the order $x_1,x_2,..,x_{n'}$ with decreasing $\delta$-value, i.e., $\delta(x_1)\geq \ldots\geq \delta(x_{n'})$. For any $p\in \{0,\ldots,n'\}$, we let $A_p$ be the set $\{x_1,\ldots,x_p\}$. Clearly $A_p$ is the best cut among all cuts $C'$ with $|C' \cap X|=p$. Now we can find the maximum weighted cut in $X\cup\{v\}$ by comparing the $n'+1$ cuts $A_0,..,A_{n'}$. Letting $\lambda$ be the value of this cut, we update $w_0(v) = \lambda$.

We can perform the same process for $v\in C$. We instead consider $A_p=\{v,x_1,\ldots,x_p\}$, and update $w_1(v)$ to the optimum value found. After having updated both weights for $v$, we can now delete all vertices in $X$.

We can apply this method to $G$ exhaustively until we are left with a graph with no edges. The desired value of the maximum weighted cut on the entire graph $G$ is the sum of the greater values of $w_0(v)$ or $w_1(v)$ for all remaining vertices $v$.

We now calculate the runtime of this method applied to one leaf-block $X$. Sorting the vertices takes $O(n'\log(n'))$ time. Since $X$ is a clique, we have $n'\log(n')\leq \frac{n'(n'+1)}{2}=m'$ for all $n'\geq 4$. We can calculate the value of the assignment $A_0$ in $O(m')$ time. Observe that the difference between cuts $A_i$ and $A_{i+1}$ for any $i \in \{0,..,n-1\}$ is in only one vertex. By only considering these local modifications we can calculate the values of the cuts $A_0,..,A_{n'}$ in $O(m')$ time. Since in every iteration we perform this process on a different block, in total we can bound our runtime with $O(n + m)$, since for blocks with $n'<4$ the runtime of $O(n'\log(n'))=O(1)$  can be charged to some vertex in the block, while for blocks with $n'\geq 4$ the runtime of $O(n'\log(n'))$ can be expressed as $O(m')$.
\end{proof}

\section{Conclusion}

With \Cref{lem:reduce,lem:onuniformCF}, our main result now follows easily:

\begin{proof}[Proof of \Cref{thm:main}]
    Given any instance $\maxcutproblem(G,\frac{w(G)}{2}+\frac{\mst(G)}{4}+\frac{k'}{4})$ with $k':=4k$, by \Cref{lem:reduce} we can in time $O(n+k\cdot m)$ either decide that the instance is a ``yes''-instance, or find a set $S\subseteq V$ with $|S|\leq 3k'=12k$ such that $G-S$ is a uniform-clique-forest. For each subset $S'\subseteq S$ we can then in time $O(n+m)$ build a \maxcutwithweightsproblem instance on the graph $G-S$, such that the vertex weights $w_0(v)$ and $w_1(v)$ of a vertex $v\in G-S$ denote the sum of the weights of edges to vertices in $S'$ and $S\setminus S'$ respectively. By \Cref{lem:onuniformCF}, each of these instances can be solved in $O(n+m)$ time. The maximum cut found in any instance given by a set $S'$ corresponds to the maximum cut $C$ of $G$ obtainable under the condition that $C\cap S=S'$. Taking into account the edges between $S$ and $S'$ and taking the maximum over all instances thus computes the maximum cut size of $G$.

    To compute the overall runtime, note that since $|S|\leq 12k$, we solve at most $2^{12k}$ \maxcutwithweightsproblem instances. Thus, the overall runtime is $O(n+k\cdot m+2^{O(k)}\cdot (n+m))=O(2^{O(k)}\cdot (n+m))$.
\end{proof}

If we want to find a cut instead of deciding the existence of a cut, we can use very similar techniques.

\begin{proof}[Proof of \Cref{thm:finding}]
    The proof of \Cref{lem:rulesaresound} is constructive: given a cut $C'$ on the reduced graph $G'$ of the assumed size, a cut $C$ on the original graph $G$ of the required size can be found in linear time in the number of removed edges and vertices. Thus, instead of applying reduction rules only until $k\leq 0$ or until the graph has no edges, we \emph{always} apply rules until the graph contains no edges. This requires at most $O(n\cdot m)$ time. Note that when we have removed all edges from the graph, the required size of a cut ($k$ larger than the \PT bound) is simply $\frac{0}{2}+\frac{0}{4}+k=k$. Thus, if $k\leq 0$ is reached, the required cut size is non-positive, thus we can start with any arbitrary cut $C'$ of the remaining independent set. We can then apply the cut extensions from the proof of \Cref{lem:rulesaresound} for all applied rules in reverse. This yields a cut of $G$ of the desired size. If otherwise we have $k>0$ when we reached a graph with no edges, we know that $|S|\leq 12k$, and we can again solve $2^{|S|}$ instances of \maxcutwithweightsproblem on $G-S$.
\end{proof}

\subsection{Open Problems}
Our result leaves a few interesting open problems. 

\proofsubparagraph{Other $\lambda$-extendible properties.}
In \cite{ptBound}, Poljak and \turzik actually not only show the lower bound for \maxcutproblem (\Cref{lem:PTbound}) but in fact they prove a very similar bound for the existence of large subgraphs fulfilling any so-called \emph{$\lambda$-extendible} property.\footnote{For \maxcutproblem this property would be bipartiteness.} Mnich, Philip, Saurabh, and Suchý~\cite{MNICH20141384} generalize the approach of \cite{blackboxFPT} for \maxcutproblem to work for a large subset of these $\lambda$-extendible properties. Note that while the title of \cite{MNICH20141384} includes ``above the \PT bound'', the authors restrict their attention to unweighted simple graphs, and thus their result applied to \maxcutproblem only implies the result of \cite{blackboxFPT}, but \emph{not} our result. We find it a very interesting direction to see if our result can be extended to also cover some more $\lambda$-extendible properties in multigraphs or positive integer-weighted graphs.

\proofsubparagraph{Kernelization.} Many previous works on \maxcutproblem parameterized above guaranteed lower bounds have also provided kernelization results~\cite{blackboxFPT,blackboxFPTlinear,aboveSpanningTree}. In particular, together with their linear-time algorithm parameterized by the distance $k$ to the \PT bound, Etscheid and Mnich~\cite{blackboxFPTlinear} also provide a linear-sized (in $k$) kernel. We are not aware of any kernelization results for \maxcutproblem on multigraphs or positive integer-weighted graphs. It would thus be very interesting to explore whether these results can also be extended to our setting.

\proofsubparagraph{FPT above better lower bounds.} Recently, Gutin and Yeo~\cite{newlowerbounds} proved new lower bounds for $\mac(G)$ for positive real-weighted graphs. In particular, they prove $\mac(G)\geq \frac{w(G)}{2}+\frac{w(M)}{2}$ where $M$ is a maximum matching of $G$, and $\mac(G)\geq \frac{w(G)}{2}+\frac{w(D)}{4}$ for any DFS-tree $D$ (which implies the \PT bound). Both of these bounds are consequences of a more general bound involving disjoint bipartite induced subgraphs, but the value of this bound is \NP-hard to compute~\cite{newlowerbounds}. The weight of the largest DFS-tree is also \NP-hard to compute~\cite{newlowerbounds}. These two bounds are thus not very suitable for an FPT algorithm, but the bound involving the maximum matching may be, since the maximum matching in a weighted graph can be computed in polynomial time using Edmonds' blossom algorithm.

\proofsubparagraph{General weights.} After going from simple graphs to multigraphs and thus positive integer-weighted graphs, it would be interesting to further generalize to positive real-weighted graphs. Here, it is not directly clear what the parameter $k$ exactly should be. Generalizing our algorithm may require completely new approaches since we cannot discretize the decrease of~$k$.

\clearpage
\bibliography{literature}

\newpage
\appendix

\section{Proof of \texorpdfstring{\Cref{lem:rulesaresound}}{Lemma Soundness}}\label{app:proofsoundness}
 We will often use the following claim that slightly strengthens the \PT bound in certain cases:

\begin{claim}\label{fact:neighboringedges}
    Let $G=(V,E)$ be a weighted graph with weights $w:E\rightarrow \N$ such that there exist edges $\{u,v\},\{v,x\}$ with $w(u,v) > w(v,x)$ and $G - \{u,v\}$ is connected. Then $G$ has a cut of size at least $\frac{w(G)}{2} + \frac{\mst(G)}{4} + \frac{1}{4}$ .
\end{claim}
\begin{proof}
    Let $G'$ = $G - \{u,v\}$. By the \PT bound we know we have a cut $C'$ of $G'$ of size at least $\frac{w(G')}{2} + \frac{\mst(G')}{4}$. We can extend this to a cut $C$ in $G$ by adding exactly one of $u$ and $v$. We choose the one such that at least half of the weight in $E(\{u,v\},V')$ goes over the cut. We have that $MSF(G') \cup \{u,v\} \cup\ \{v,x\}$ is a spanning forest of $G$, therefore $\mst(G') + w(u,v) + w(v,x) \geq \mst(G)$. The cut $C$ has weight at least
    \begin{align*}
        w(C)&\geq \frac{w(G')}{2} + \frac{\mst(G')}{4} + \frac{w(\{u,v\},V')}{2} + w(u,v)\\
        &= \frac{w(G)}{2} + \frac{\mst(G')}{4} + \frac{w(u,v)}{2}\\
        &\geq \frac{w(G)}{2} + \frac{\mst(G')}{4} + \frac{w(u,v)}{4} + \frac{w(v,x)+1}{4}\\
        &\geq \frac{w(G)}{2} + \frac{\mst(G)}{4} + \frac{1}{4}.\qedhere
    \end{align*}
\end{proof}

Let us now prove \Cref{lem:rulesaresound}.

\begin{proof}[Proof of \Cref{lem:rulesaresound}]
We first see that each rule preserves connectedness simply by their preconditions. Each rule either explicitly requires that the resulting graph is connected (\Cref{rule:edge,rule:triplet,rule:specialxy}), or removes a whole leaf-block of $G$, except for the cut vertex (\Cref{rule:vertexclique,rule:specialvertexmark,rule:vertexmark,rule:vertexclique_constant,rule:specialtriplet}). %

We now prove the required cut size implication for each rule independently.
We need to prove that if there exists a cut $C'$ in $G'$ that produces a cut of size $\frac{w(G')}{2} + \frac{\mst(G')}{4} + \frac{k'}{4}$, then this can be extended to a cut $C$ of $G$ of size $\frac{w(G)}{2} + \frac{\mst(G)}{4} + \frac{k}{4}$. We thus assume that such a cut $C'$ exists, and then extend it in such a way that $C\cap V' = C'$.
We perform a case distinction on the rule that we applied to $G$ to obtain $G'$. Recall that for all rules except \Cref{rule:vertexclique_constant}, $k'=k-1$.

\proofsubparagraph{\Cref{rule:edge}:} We extend $C'$ by putting $x$ and $y$ on different sides of the cut. Among the two possibilities, we choose the one such that at least half the weight in $E(\{x,y\},V')$ goes over the cut. We get a cut of size at least
\begin{align*}w(C')&\geq \frac{w(G')}{2} + \frac{\mst(G')}{4} + \frac{k'}{4} + w(x,y) + \frac{w(E(\{x,y\},V'))}{2} \\
&= \frac{w(G)}{2} + \frac{\mst(G')}{4} + \frac{k'}{4} + \frac{w(x,y)}{2}\\
&\geq \frac{w(G)}{2} + \frac{\mst(G')}{4} + \frac{k'}{4} + \frac{w(x,y)}{4} + \frac{w(y,z)+1}{4}
\end{align*}
We now see that $\mst(G)\leq \mst(G')+w(x,y)+w(y,z)$, and we thus get 
\[w(C')\geq \frac{w(G)}{2} + \frac{\mst(G)}{4}+\frac{k}{4}.\]

\proofsubparagraph{\Cref{rule:vertexclique_constant}:} We can assume without loss of generality that $v\in C'$. Let $n' := |X \cup \{v\}|$. Observe that the sum of the total weight in $G[X\cup\{v\}]$ is $c(\frac{n'(n'-1)}{2})$ for the integer $c$ such that all edges in $G[X\cup\{v\}]$ have weight $c$. 
If $n'$ is odd, $|X|$ is even, and we can add exactly half the vertices to $C$. This way we have a cut $C''$ in $G[X \cup \{v\}]$ of size at least 
\begin{align*}
    w(C'')&\geq c\Big(\frac{n'+1}{2}\Big)\Big(\frac{n'-1}{2}\Big)\\
    &= c\Big(\frac{n'(n'-1)}{4}\Big) + c\Big(\frac{n'-1}{4}\Big)\\
    &= \frac{w(G[X \cup \{v\}])}{2} + \frac{\mst(G[X \cup \{v\}])}{4}.
\end{align*}
If $n'$ is even, we add $\frac{n'}{2}$ of the vertices of $X$ to $C$, and leave $\frac{n'}{2}-1$ vertices out of $C$. In this case, we have a cut $C''$ in $G[X \cup \{v\}]$ of size at least
\begin{align*}
    w(C'')&\geq c\Big(\frac{n'}{2}\Big)\Big(\frac{n'}{2}\Big)\\
    &= c\Big(\frac{n'(n'-1)}{4}\Big) + c\Big(\frac{n'}{4}\Big)\\
    &\geq \frac{w(G[X \cup \{v\}])}{2} + \frac{\mst(G[X \cup \{v\}])}{4}.
\end{align*}
In either case, we can see that we can combine $C'$ and $C''$ to a cut $C$ in $G$ of size at least
\begin{align*}
    w(C)&\geq \frac{w(G[X \cup \{v\}])}{2} + \frac{\mst(G[X \cup \{v\}])}{4} + \frac{w(G')}{2} + \frac{\mst(G')}{4} + \frac{k'}{4}\\
    &= \frac{w(G)}{2} + \frac{\mst(G)}{4} + \frac{k}{4},
\end{align*}
where in the last equality we used that $k'=k$ for this rule.

\proofsubparagraph{\Cref{rule:vertexclique}:} Since $G[X\cup \{v\}]$ is not uniform, we can apply \Cref{fact:neighboringedges} to $G[X \cup \{v\}]$ to obtain a cut $C''$ in $G[X \cup \{v\}]$ of size at least $\frac{w(G[X \cup \{v\}])}{2} + \frac{\mst(G[X \cup \{v\}])}{4} + \frac{1}{4}$. We now assume without loss of generality that $v\in C''\Leftrightarrow v\in C'$, i.e., both $C'$ and $C''$ put $v$ on the same side of the cut. In this case we can combine $C'$ and $C''$ to a cut $C$ of size at least 
\begin{align*}
    w(C)&\geq \frac{w(G[X \cup \{v\}])}{2} + \frac{\mst(G[X \cup \{v\}])}{4} + \frac{1}{4} + \frac{w(G')}{2} + \frac{\mst(G')}{4} + \frac{k'}{4}\\
    &= \frac{w(G)}{2} + \frac{\mst(G)}{4} + \frac{k}{4}.
\end{align*}

\proofsubparagraph{\Cref{rule:vertexmark}:} We know that $v$ must be adjacent to more than $1$ and less than $|X|$ vertices of $X$. We first do a case distinction on whether $G[X \cup \{v\}]$ is uniform or not.

If $G[X \cup \{v\}]$ is not uniform, we use the same argument as for the previous rule. Let $y \in X$ be a vertex not adjacent to $v$. Observe that for any $x\in X$ such that $\{v,x\}\in E$, $G[X \cup \{v\} - \{v,x\}]$ and $G[X \cup \{v\} - \{x,y\}]$ are both connected. Since $G[X\cup\{v\}]$ is not uniform but $G[X]$ is, we can find such an $x$ such that either $w(x,y) > w(x,v)$ or $w(x,y) < w(x,v)$. Therefore, we can use \Cref{fact:neighboringedges} on $G[X \cup \{v\}]$. This gives us a cut $C''$ in $G[X\cup\{v\}]$ of size at least $\frac{w(G[X \cup \{v\}])}{2} + \frac{\mst(G[X \cup \{v\}])}{4} + \frac{1}{4}$. Combining this cut with $C'$, we get a cut $C$ in $G$ of size at least 
\begin{align*}
    w(C)&\geq \frac{w(G')}{2} + \frac{\mst(G')}{4} + \frac{k'}{4} + \frac{w(G[X \cup \{v\}])}{2} + \frac{\mst(G[X \cup \{v\}])}{4} + \frac{1}{4} \\
&\geq \frac{w(G)}{2} + \frac{\mst(G)}{4} + \frac{k}{4}.
\end{align*}

Otherwise $G[X \cup \{v\}]$ is $c$-uniform. Let $m'$ = $w(G[X\cup\{v\}])$ and $n'=|X|$. We can order the vertices in $X$ as $x_1, x_2, \ldots, x_{n'}$ such that $v$ is adjacent to exactly $x_1,\ldots,x_r$, but not $x_{r+1},\ldots,x_{n'}$. Assume without loss of generality that $v\in C'$. We add $v$ and all $x_i$ for $i >\lceil\frac{n'}{2}\rceil$ to a cut $C''$ of $G[X\cup\{v\}]$. This cut has size $s:=c(\lceil \frac{n'}{2} \rceil\cdot\lfloor \frac{n'}{2} \rfloor + min\{r, \lceil \frac{n'}{2} \rceil\})$. Note that $m' = c(\frac{n'(n'-1)}{2} + r)$, thus we can rephrase $s=\frac{m'}{2} + c(\frac{n'}{4} - \frac{n'^2}{4} + (\lceil \frac{n'}{2} \rceil)(\lfloor \frac{n'}{2} \rfloor) + min\{\frac{r}{2}, \lceil \frac{n'}{2} \rceil - \frac{r}{2}\})$.
\\ If $n'$ is even, $(\lceil \frac{n'}{2} \rceil)(\lfloor \frac{n'}{2} \rfloor) = \frac{n'^2}{4}$, and then $s\geq\frac{m'}{2} + c\frac{n'}{4} + \frac{c}{2} \geq \frac{m'}{2} + c\frac{n'}{4} + \frac{1}{4}$.
\\ If $n'$ is odd, $(\lceil \frac{n'}{2} \rceil)(\lfloor \frac{n'}{2} \rfloor) = \frac{(n'+1)(n'-1)}{4} = \frac{n'^2}{4} - \frac{1}{4}$, and then $s\geq \frac{m'}{2} + c\frac{n'}{4} + \frac{c}{2} - \frac{c}{4} \geq \frac{m'}{2} + c\frac{n'}{4} + \frac{1}{4}$, as well.
\\ In either case we can combine $C''$ on $G[X\cup\{v\}]$ and $C'$ on $G'$ to get a cut $C$ of $G$ of size at least 
\begin{align*}
    w(C)&\geq \frac{w(G')}{2} + \frac{\mst(G')}{4} + \frac{k'}{4} + \frac{m'}{2} + c\frac{n'}{4} + \frac{1}{4}\\
    &\geq \frac{w(G)}{2} + \frac{\mst(G)}{4} + \frac{k}{4},
\end{align*}
where we used that an MSF of $G'$ can be turned into a spanning forest of $G$ by adding $n'$ edges of weight $c$.

\proofsubparagraph{\Cref{rule:specialvertexmark}:} Let $X' = X - \{x,y\}$. If $w(x,v) > c$ or $w(y,v) > c$, since $G[X \cup \{v\} - \{v,x\}]$ and $G[X \cup \{v\} - \{v,y\}]$ are connected, we know by  \Cref{fact:neighboringedges} that $G[X \cup \{v\}]$ has a cut $C''$ of size at least $\frac{w(G[X \cup \{v\}])}{2} + \frac{\mst(G[X \cup \{v\}])}{4} + \frac{1}{4}$. Since $G'$ and $G[X\cup\{v\}]$ overlap in only one vertex $v$ we can w.l.o.g. assume that $v\in C''\Leftrightarrow v\in C'$, and we can combine $C''$ and $C'$ to a cut $C$ of $G$ of size at least 
\begin{align*}
    w(C)&\geq \frac{w(G')}{2} + \frac{\mst(G')}{4} + \frac{k'}{4} +\frac{w(G[X \cup \{v\}])}{2} + \frac{\mst(G[X \cup \{v\}])}{4} + \frac{1}{4} \\
    &\geq \frac{w(G)}{2} + \frac{\mst(G)}{4} + \frac{k}{4}.
\end{align*}

Thus, from now on we may assume $w(x,v) =w(y,v)= c$, i.e., the only edge in $G[X\cup\{v\}]$ that does not have weight $c$ is the edge $\{x,y\}$ of weight $>c$. For the remaining cases, we perform a case distinction over the size of $X'$. Without loss of generality we assume that $v\not\in C'$.
\begin{itemize}
    \item Case 1: $|X'|=1$. Let $u$ be the only vertex in $X'$. Observe $\mst(G') + w(x,v) + w(y,v) + w(u,x) \geq \mst(G)$.
\begin{itemize}
    \item Assume $w(x,y) > 2c$. We create $C$ by adding $x$ to $C'$. Then $C$ has size at least
    \begin{align*}
        w(C)&\geq \frac{w(G')}{2} + \frac{\mst(G')}{4} + \frac{k'}{4} + w(x,y) + w(x,v) + w(x,u)\\
        &\geq \frac{w(G)}{2} + \frac{\mst(G')}{4} + \frac{k'}{4} + \frac{w(x,y) + w(x,v)+w(x,u)-w(y,v)-w(y,u)}{2}\\
        & =\frac{w(G)}{2} + \frac{\mst(G')}{4} + \frac{k'}{4} + \frac{w(x,y)}{2}\\
        & >\frac{w(G)}{2} + \frac{\mst(G')}{4} + \frac{k'}{4} + \frac{2c}{2}\\
        &\geq \frac{w(G)}{2} + \frac{\mst(G')}{4} + \frac{k'}{4} + \frac{3c}{4} + \frac{1}{4}\\
        &\geq \frac{w(G)}{2} + \frac{\mst(G)}{4} + \frac{k}{4}.
    \end{align*}

    \item Assume $w(x,y) \leq 2c$. We create $C$ by adding $x$ and $y$ to $C'$. Then $C$ has size at least 
    \begin{align*}
        w(C)&\geq \frac{w(G')}{2} + \frac{\mst(G')}{4} + \frac{k'}{4} + w(y,v) + w(y,u) + w(x,v) + w(x,u) \\
        &\geq \frac{w(G')}{2} + \frac{\mst(G')}{4} + \frac{k'}{4} + 3c + \frac{w(x,y)}{2}\\
        &\geq \frac{w(G)}{2} + \frac{\mst(G')}{4} + \frac{k'}{4} + c\\
        &\geq \frac{w(G)}{2} + \frac{\mst(G)}{4} + \frac{k'}{4} + \frac{c}{4}\\
        &\geq \frac{w(G)}{2} + \frac{\mst(G)}{4} + \frac{k}{4}.
    \end{align*}
\end{itemize}
\item Case 2: $|X'|=:n'>1$.
Observe $w(G) = w(G') + c(\frac{n'(n'-1)}{2} + 2n'+2) + w(x,y)$ and $\mst(G') + c(n'+2) \geq \mst(G)$.
\begin{itemize}
    \item Assume $w(x,y) \geq 2c$. We start with a cut on $G[X']$ of size at least $\frac{w(G[X'])}{2} + \frac{\mst(G[X'])}{4} = \frac{w(G[X'])}{2} + c(\frac{n'-1}{4})$ as guaranteed by the \PT bound. Then we extend this to a cut on $G[X'\cup\{x,y\}]$ by adding exactly one of $x$ and $y$, choosing of the two possibilities the one that cuts at least half the weight in $E(X', \{x,y\})$. We combine this cut with the cut $C'$.
    
    The resulting cut $C$ has size at least 
    \begin{align*}
        w(C)&\geq \frac{w(G')}{2} + \frac{\mst(G')}{4} + \frac{k'}{4} + \frac{w(G[X'])}{2} + c\Big(\frac{n'-1}{4}\Big) + \frac{|E(X', \{x,y\})|}{2} + w(x,y) + c \\
        &\geq \frac{w(G)}{2} + \frac{\mst(G')}{4} + \frac{k'}{4} + c\Big(\frac{n'-1}{4}\Big)  + \frac{w(x,y)}{2} \\
        &\geq \frac{w(G)}{2} + \frac{\mst(G')}{4} + \frac{k'}{4} + c\Big(\frac{n'+3}{4}\Big)\\
        &\geq \frac{w(G)}{2} + \frac{\mst(G)}{4} + \frac{k'}{4} + \frac{c}{4}\\
        &\geq \frac{w(G)}{2} + \frac{\mst(G)}{4} + \frac{k}{4}.
    \end{align*}
    \item Assume $w(x,y) < 2c$. We add both $x$ and $y$ to the cut $C'$.
    \\If $n'$ is odd we add $\frac{n'-1}{2}$ vertices of $X'$ to $C'$. The resulting cut $C$ has size at least 
    \begin{align*}
        w(C)&\geq \frac{w(G')}{2} + \frac{\mst(G')}{4} + \frac{k'}{4} + c\Bigg(\Big(\frac{n'+3}{2}\Big)\Big(\frac{n'+1}{2}\Big)+2\Bigg)\\
        &\geq \frac{w(G')}{2} + \frac{\mst(G')}{4} + \frac{k'}{4} + c\Big(\frac{n'^2+4n'+3}{4}+1\Big) + \frac{w(x,y)+1}{2}\\
        &=\frac{w(G)}{2} + \frac{\mst(G')}{4} + \frac{k'}{4} + c\Big(\frac{n'-1}{4}+1\Big) + \frac{1}{2}\\
        &= \frac{w(G)}{2} + \frac{\mst(G')}{4} + \frac{k'}{4} + c\Big(\frac{n'+3}{4}\Big) + \frac{1}{2}\\
        &\geq \frac{w(G)}{2} + \frac{\mst(G)}{4} + \frac{k'}{4} + \frac{1}{2}\\
        &\geq \frac{w(G)}{2} + \frac{\mst(G)}{4} + \frac{k}{4}.
    \end{align*}  
    \\If $n'$ is even we add $\frac{n'}{2}-1$ vertices of $X'$ to 
$C'$. The resulting cut $C$ has size at least
    \begin{align*}
        w(C)&\geq \frac{w(G')}{2} + \frac{\mst(G')}{4} + \frac{k'}{4} + c\Bigg(\Big(\frac{n'+2}{2}\Big)\Big(\frac{n'+2}{2}\Big)+2\Bigg),
    \end{align*}
    which is strictly larger than in the $n'$ odd case.
\end{itemize}
\end{itemize}

\proofsubparagraph{\Cref{rule:triplet}:} Let $\min=\min (V \setminus \{a,b,c\},\{a,b,c\})$, and let $e_{\min}$ be an edge of weight $\min$ in $E(V\setminus\{a,b,c\},\{a,b,c\})$. Observe that $MSF(G')$ together with $e_{\min},\{a,b\}$, and $\{b,c\}$ forms a spanning forest of $G$. Therefore $\mst(G') + \min + w(a,b) + w(b,c) = \mst(G') + \min + 2w(a,b) \geq \mst(G)$.

We consider two subsets of $\{a,b,c\}$: $A_1=\{a,c\}$, and $A_2=\{b\}$. Considering these as cuts of $G$, both cuts cut the edges $\{a,b\}$ and $\{b,c\}$, and at least one of these cuts gets at least half of the total weight in $E(V(G'),\{a,b,c\})$. Enhancing $C'$ by that set, we therefore get a cut $C$ of $G$ of size at least 
\begin{align*}
    w(C)&\geq \frac{w(G')}{2} + \frac{\mst(G')}{4} + \frac{k'}{4} + \frac{w(V(G'),\{a,b,c\})}{2} + w(a,b) + w(b,c) \\
    &\geq \frac{w(G)}{2} + \frac{\mst(G')}{4} + \frac{k'}{4} + \frac{w(a,b)}{2} + \frac{w(b,c)}{2}\\
    &\geq \frac{w(G)}{2} + \frac{\mst(G')}{4} + \frac{k'}{4} + \frac{w(a,b)}{2} + \frac{\min+1}{4}\\
    &\geq \frac{w(G)}{2} + \frac{\mst(G)}{4} + \frac{k'}{4} + \frac{1}{4}\\
    &\geq \frac{w(G)}{2} + \frac{\mst(G)}{4} + \frac{k}{4}.
\end{align*}

\proofsubparagraph{\Cref{rule:specialtriplet}:} If $b$ is adjacent to $v$ we can augment $C'$ by adding $a,b,c$ to $C$ if and only if $v\not\in C'$. Observe $MSF(G') \cup \{a,v\} \cup \{b,v\} \cup \{c,v\}$ is a spanning forest of $G$. Also by the conditions of \Cref{rule:specialtriplet} we have $\frac{w(a,v)}{4} + \frac{w(b,v)}{4} \geq \frac{w(a,b)}{2} + \frac{w(a,b)}{2} = \frac{w(a,b)}{2} + \frac{w(b,c)}{2}$. We can thus analyze the cut $C$ to have size at least
\begin{align*}
    w(C)&\geq \frac{w(G')}{2} + \frac{\mst(G')}{4} + \frac{k'}{4} + w(a,v) + w(b,v) + w(c,v) \\
    &\geq \frac{w(G')}{2} + \frac{\mst(G')}{4} + \frac{k'}{4} + \frac{3w(a,v)}{4} + \frac{3w(b,v)}{4} + \frac{w(a,b)}{2} + \frac{w(b,c)}{2} + w(c,v)\\
    &\geq \frac{w(G)}{2} + \frac{\mst(G')}{4} + \frac{k'}{4} + \frac{w(a,v)}{4} + \frac{w(b,v)}{4} + \frac{w(c,v)}{2}\\
    &\geq \frac{w(G)}{2} + \frac{\mst(G)}{4} + \frac{k'}{4} + \frac{w(c,v)}{4}\\
    &\geq \frac{w(G)}{2} + \frac{\mst(G)}{4} + \frac{k}{4}.
\end{align*}
 
If $b$ is not adjacent to $v$, add $a,c$ to $C$ if and only if $v\not\in C'$, and we add $b$ to $C$ if and only if $v\in C'$. Thus the edges $\{a,b\},\{b,c\},\{a,v\},\{c,v\}$ are all cut. Note that $MSF(G')\cup\{c,v\}\cup\{a,b\}\cup\{b,c\}$ is a spanning forest of $G$. The cut $C$ has size at least
\begin{align*}
    w(C)&\geq \frac{w(G')}{2} + \frac{\mst(G')}{4} + \frac{k'}{4} + w(a,v) + w(a,b) + w(b,c) + w(c,v)\\
    &\geq \frac{w(G)}{2} + \frac{\mst(G)}{4} + \frac{k'}{4} + \frac{w(a,v)}{2} + \frac{w(c,v)}{4} + \frac{w(a,b)}{4} + \frac{w(b,c)}{4}\\
    &\geq \frac{w(G)}{2} + \frac{\mst(G)}{4} + \frac{k}{4}.     
\end{align*}

\proofsubparagraph{\Cref{rule:specialxy}:} Let $v$ be the only neighbor of $\{x,y\}$ in $Y$ and let $\overline{n} = |X|$. We first extend $C'$ to $C''$ by adding $x,y$ to $C''$ if and only if $v\not\in C'$. We then extend $C''$ to $C$ as follows.

Without loss of generality, assume $x,y\not\in C''$. We perform a case distinction on the parity of $\overline{n}$. Note that $w(G[X \cup \{x,y\}]) = c(\frac{\overline{n}(\overline{n}-1)}{2} + 2\overline{n})$.
\begin{itemize}
    \item If $\overline{n}$ is odd, we add $\frac{\overline{n}+1}{2}$ of the vertices in $X$ to $C$. In $G[X\cup\{x,y\}]$ this cuts in total a weight of 
    \[c\Bigg(\Big(\frac{\overline{n}+1}{2}\Big)\Big(\frac{\overline{n}-1}{2}\Big) + 2\frac{\overline{n}+1}{2}\Bigg)= c\Big(\frac{\overline{n}(\overline{n}-1)}{4} + \frac{\overline{n}-1}{4} + \overline{n} + 1\Big) = \frac{w(G[X \cup \{x,y\}])}{2} + c\Big(\frac{\overline{n}}{4} + \frac{3}{4}\Big).\]
    \item If $\overline{n}$ is even, we add $\frac{\overline{n}}{2} + 1$ vertices in $X$ to $C$. In $G[X\cup\{x,y\}]$ this cuts in total a weight of  \[c\Bigg(\Big(\frac{\overline{n}}{2} + 1\Big)\Big(\frac{\overline{n}}{2}-1\Big) + 2\Big(\frac{\overline{n}}{2}+1\Big)\Bigg) = c\Big(\frac{\overline{n}^2}{4} + \overline{n} + 1\Big) = c\Big(\frac{\overline{n}^2}{4} + \frac{3}{4}\overline{n} + \frac{\overline{n}}{4} + 1\Big) = \frac{w(G[X \cup \{x,y\}])}{2} + c\Big(\frac{\overline{n}}{4} + 1\Big).\]
\end{itemize}
In either case we thus have that $C$ cuts at least half of the weight in $G[X\cup\{x,y\}]$ plus $c(\frac{\overline{n}}{4}+\frac{3}{4})$.

Observe that $ \mst(G) \leq \mst(G') + c\overline{n} + w(x,v) + w(y,v)$. In total we can thus bound the size of the cut $C$ as
\begin{align*}
    w(C)&\geq \frac{w(G')}{2} + \frac{\mst(G')}{4} + \frac{k'}{4} + \frac{w(G[X \cup \{x,y\}])}{2} + c(\frac{\overline{n}}{4} + \frac{3}{4}) + w(x,v) + w(y,v)\\
    &= \frac{w(G)}{2} + \frac{\mst(G')}{4} + \frac{k'}{4} + c(\frac{\overline{n}}{4} + \frac{3}{4}) + \frac{w(x,v)}{2} + \frac{w(y,v)}{2}\\
    &\geq \frac{w(G)}{2} + \frac{\mst(G)}{4} + \frac{k}{4}.
\end{align*}

We conclude that for every rule, from a cut $C'$ of $G'$ of the guaranteed size we can build a cut $C$ of $G$ of the required size, and thus the lemma follows.
\end{proof}

\end{document}